\documentclass[a4paper,english,cleveref,autoref,final]{lipics-v2019}

\newif\ifproc
\procfalse

\nolinenumbers

\usepackage{xspace}
\usepackage{bm}
\usepackage{hyperref}
\usepackage{csquotes}
\usepackage{color}
\usepackage{microtype}
\usepackage{nicefrac}
\usepackage{mathtools}
\usepackage[inline]{enumitem}
\setlist[enumerate]{label={\arabic*)},font={\bfseries}}
\usepackage{etoolbox}
\BeforeBeginEnvironment{enumerate*}{%
	\setlist[enumerate,1]{label=(\alph*), font={\bfseries}}
}
\AfterEndEnvironment{enumerate*}{%
	\setlist[enumerate,1]{label=\arabic*.}
}

\usepackage[sort&compress,numbers]{natbib}

\usepackage{tikz}
\usetikzlibrary{fit,patterns,decorations.pathreplacing}

\newcommand{\Oh}{\mathcal{O}}
\newcommand{\OhOp}[1]{\Oh\mathopen{}\mathclose\bgroup\left( #1 \aftergroup\egroup\right)}
\DeclareMathOperator{\poly}{{\rm poly}}

\usepackage{xspace}
\newcommand{\FPT}{{\sf FPT}\xspace}
\newcommand{\NP}{{\sf NP}\xspace}
\newcommand{\NPh}{\hbox{{\sf NP}-hard}\xspace}

\newcommand{\XP}{{\sf XP}\xspace}

\newcommand{\Wh}[1]{$\mathsf{W[#1]}$-hard\xspace}
\newcommand{\Whness}[1]{$\mathsf{W[#1]}$-hardness\xspace}

\usepackage{tabularx}
\newcommand{\prob}[3]{
\begin{center}
\begin{tabularx}{\textwidth}{lX}
	\multicolumn{2}{l}{#1}\\
	{\bf Input:}&{#2}\\
	{\bf Find:}&{#3}
\end{tabularx}
\end{center}
}

\usepackage[textwidth=20mm,obeyFinal,colorinlistoftodos,disable]{todonotes}
\newcommand{\mytodo}[2]{\todo[size=\tiny, color=#1!50!white]{#2}\xspace}
\newcommand{\myinlinetodo}[2]{\todo[size=\small, color=#1!50!white, inline]{#2}\xspace}
\newcommand{\jzcom}[1]{\mytodo{green}{#1}}

\newcommand{\mkcom}[1]{\mytodo{blue}{#1}}
\newcommand{\mkinline}[1]{\myinlinetodo{blue}{#1}}

\newcommand{\conv}{\textrm{conv}}

\newcommand{\CC}{\mathcal{C}}

\DeclarePairedDelimiter\ceil{\lceil}{\rceil}
\DeclarePairedDelimiter\floor{\lfloor}{\rfloor}

\DeclareMathOperator{\rank}{rank}

\usepackage{stackengine}
\stackMath

\newcommand{\df}{:=}

\def\ve#1{\mathchoice{\mbox{\boldmath$\displaystyle\bf#1$}}
{\mbox{\boldmath$\textstyle\bf#1$}}
{\mbox{\boldmath$\scriptstyle\bf#1$}}
{\mbox{\boldmath$\scriptscriptstyle\bf#1$}}}

\newcommand\veb{{\ve b}}
\newcommand\vecc{{\ve c}}

\newcommand\vel{{\ve l}}
\newcommand\ven{{\ve n}}
\newcommand\vem{{\ve m}}

\newcommand\vep{{\ve p}}

\newcommand\ves{{\ve s}}

\newcommand\veu{{\ve u}}

\newcommand\vew{{\ve w}}
\newcommand\vex{{\ve x}}
\newcommand\vey{{\ve y}}

\newcommand\vezero{{\ve 0}}

\def\R{\mathbb{R}}
\def\Z{\mathbb{Z}}
\def\N{\mathbb{N}}
\def\Q{\mathbb{Q}}

\usepackage{etoolbox}

\newcommand{\lv}[1]{}

\newcommand{\appendixText}{}

\newcommand{\toappendix}[1]{\gappto{\appendixText}{{#1}}}

\newcommand{\appmark}{$\star$}

\theoremstyle{theorem}
\makeatletter
\newtheorem*{rep@theorem}{\rep@title}
\newcommand{\newreptheorem}[2]{%
\newenvironment{rep#1}[1]{%
 \def\rep@title{#2 \ref{##1}}%
 \begin{rep@theorem}}%
 {\end{rep@theorem}}}
\makeatother

\newcommand{\LL}{\mathcal{L}}

\newreptheorem{theorem}{Theorem}
\newreptheorem{lemma}{Lemma}
\newreptheorem{corollary}{Corollary}

\newcommand{\binpacking}{\textsc{Bin Packing}\xspace}
\newcommand{\balancedbinpacking}{\textsc{Balanced Bin Packing}\xspace}
\newcommand{\unarybinpacking}{\textsc{Unary Bin Packing}\xspace}
\newcommand{\unarybalancedbinpacking}{\textsc{Unary Balanced Bin Packing}\xspace}

\title{Complexity of Scheduling Few Types of Jobs on Related and Unrelated Machines} 

\titlerunning{}

\author{Martin Koutecký}{Computer Science Institute, Charles University, Czech Republic}{koutecky@iuuk.mff.cuni.cz}{https://orcid.org/0000-0002-7846-0053}{}

\author{Johannes Zink}{Institut f\"ur Informatik, Universit\"at W\"urzburg, Germany}{johannes.zink@uni-wuerzburg.de}{https://orcid.org/0000-0002-7398-718X}{}

\authorrunning{M.\ Koutecký and J.\ Zink}

\Copyright{Martin Koutecký and Johannes Zink}



\category{}

\supplement{}

\acknowledgements{
We thank the organizers of the HOMONOLO 2019 workshop for providing a warm and stimulating research environment, which gave birth to the initial ideas of this paper.
We also thank the anonymous reviewers for their helpful remarks.
M.~Koutecký was partially supported by Charles University project UNCE/SCI/004 and by the project 19-27871X of GA ČR.
}

\ifproc
\relatedversion{TODO A full version of the paper is available at \url{https://arxiv.org/abs/...}.}
\fi

\ifproc\else
\hideLIPIcs  
\fi

\EventEditors{Yixin Cao, Siu-Wing Cheng, and Minming Li}
\EventNoEds{3}
\EventLongTitle{31st International Symposium on Algorithms and Computation (ISAAC 2020)}
\EventShortTitle{ISAAC 2020}
\EventAcronym{ISAAC}
\EventYear{2020}
\EventDate{December 14--18, 2020}
\EventLocation{Hong Kong, China (Virtual Conference)}
\EventLogo{}
\SeriesVolume{181}
\ArticleNo{49}

\begin{document}

\maketitle

\begin{abstract}
The task of scheduling jobs to machines while minimizing the total makespan, the sum of weighted completion times, or a norm of the load vector, are among the oldest and most fundamental tasks in combinatorial optimization.
Since all of these problems are in general \NPh, much attention has been given to the regime where there is only a small number $k$ of job types, but possibly the number of jobs $n$ is large; this is the few job types, high-multiplicity regime.
Despite many positive results, the hardness boundary of this regime was not understood until now.

We show that makespan minimization on uniformly related machines ($Q|HM|C_{\max}$) is \NPh already with $6$ job types, and that the related \textsc{Cutting Stock} problem is \NPh already with $8$ item types.
For the more general unrelated machines model ($R|HM|C_{\max}$), we show that if either the largest job size $p_{\max}$, or the number of jobs $n$ are polynomially bounded in the instance size $|I|$, there are algorithms with complexity $|I|^{\poly(k)}$.
Our main result is that this is unlikely to be improved, because $Q||C_{\max}$ is \Wh{1} parameterized by $k$ already when $n$, $p_{\max}$, and the numbers describing the speeds are polynomial in $|I|$; the same holds for $R|HM|C_{\max}$ (without speeds) when the job sizes matrix has rank $2$.
Our positive and negative results also extend to the objectives $\ell_2$-norm minimization of the load vector and, partially, sum of weighted completion times $\sum w_j C_j$.

Along the way, we answer affirmatively the question whether makespan minimization on identical machines ($P||C_{\max}$) is fixed-parameter tractable parameterized by $k$, extending our understanding of this fundamental problem.
Together with our hardness results for $Q||C_{\max}$ this implies that the complexity of $P|HM|C_{\max}$ is the only remaining open case.
\end{abstract}



\section{Introduction}
Makespan minimization is arguably the most natural and most studied scheduling problem: in the parallel machines model, we have $m$ machines, $n$ jobs with sizes $p_1, \dots, p_n$, and the task is to assign them to machines such that the sum of sizes of jobs on any machine is minimized.
Seen differently, this is the (decision version of the) \textsc{Bin Packing} problem: can a set of items be packed into a given number of bins?
\textsc{Bin Packing} is \NPh, so it is natural to ask which restrictions  make it polynomial time solvable.
Say there are only $k$ distinct item sizes $p_1, \dots, p_k$, and so the items are given by a vector of multiplicities $n_1, \dots, n_k$ with $n = \sum_{j=1}^k n_j$; let $p_{\max} = \max_j p_j$.
Goemans and Rothvoss~\cite{GoemansRothvoss2014} showed that \textsc{Bin Packing} can be solved in time $(\log p_{\max})^{f(k)} \poly\log n$ for some function $f$.\footnote{The complexity stated in~\cite{GoemansRothvoss2014} is $(\log \max{C_{\max},n})^{f(k)} \poly \log n$, but a close inspection of their proof reveals that \textbf{a)} the dependence on $n$ is unnecessary, and \textbf{b)} it is possible to use a better bound on the number of vertices of a polytope and obtain the complexity stated here.}
Note that makespan minimization is polynomial when $k$ is fixed by simple dynamic programming; the difficult question is whether it is still polynomial in the \emph{high-multiplicity} setting where jobs are encoded by the multiplicity vector $\ven = (n_1, \dots, n_k)$.
By the equivalence with scheduling, Goemans and Rothvoss showed that high-multiplicity makespan minimization on identical machines is polynomial if the number of job types $k$ is fixed.

Since 2014, considerable attention has been given to studying the complexity of various scheduling problems in the regime with few job types~\cite{JansenKlein2017,KnopKLMO:2019,KnopK:2017,KnopK2020,ChenMYZ:2018,HermelinKPS:2018,HermelinPST:2019,MnichW:2015,Jansen:2017}, and similar techniques have been used to obtain approximation algorithms~\cite{JansenKMR:2018,Levin:2019,JansenLM:2019}.
However, any answer to the following simple and natural question was curiously missing: 
\begin{quote}
\emph{What is the most restricted machine model in which high-multiplicity makespan minimization becomes \NPh, even when the number of job types is fixed?}
\end{quote}
There are three main machine models in scheduling: identical, uniformly related, and unrelated machines.
In the uniformly related machines model, machine~$M_i$ (for $i \in [m]$) additionally has a \emph{speed $s_i$}, and processing a job of size $p_j$ takes time $p_j/s_i$ on such a machine.
In the unrelated machines model, each machine~$M_i$ (for $i \in [m]$) has its own vector of job sizes $\vep^i = (p^i_1, \dots, p^i_k)$, so that $p^i_j$ is the time to process a job of type $j$ on machine~$M_i$.
The makespan minimization problem in the identical, uniformly related, and unrelated machines model is denoted shortly as $P||C_{\max}$, $Q||C_{\max}$, and $R||C_{\max}$~\cite{LawlerEtAl1993}, respectively, with the high-multiplicity variant being $P|HM|C_{\max}$ and analogously for the other models.
Notice that the job sizes matrix $\vep$ of a $Q||C_{\max}$ instance is of rank $1$: the vector $\vep^i$ for machine $M_i$ is simply $\vep' / s_i$ for $\vep' = (p_1, \dots, p_k)$, and $\vep = \vep' \cdot (1/\ves)^{\intercal}$ for the speeds vector $\ves = (s_1, \dots, s_m)$.
Hence, the rank of the job sizes matrix has been studied~\cite{BhaskaraKTW:2013,ChenMYZ:2018,ChenJZ:2018} as a helpful measure of complexity of an $R||C_{\max}$ instance: intuitively, the smaller the rank, the closer is the instance to $Q||C_{\max}$.
We answer the question above:
\begin{theorem}
	\label{thm:QCmaxNPhard}
	$Q|HM|C_{\max}$ is \NPh already for 6 job types.
\end{theorem}
The \textsc{Cutting Stock} problem relates to \textsc{Bin Packing} in the same way as $Q||C_{\max}$ relates to $P||C_{\max}$: instead of having all bins have the same capacity, there are now several bin types with a different capacity and cost, and the task is to pack all items into bins of minimum cost.
\textsc{Cutting Stock} is a famous and fundamental problem whose study dates back to the ground-breaking work of Gilmore and Gomory~\cite{GilmoreGomory1961}.
It is thus surprising that the natural question whether \textsc{Cutting Stock} with a fixed number of item types is polynomial or \NPh has not been answered until now:
\begin{theorem} \label{thm:cuttingstock}
\textsc{Cutting Stock} is \NPh already with $8$ item types.
\end{theorem}

\vspace{-1em}
\subparagraph*{Parameterized Complexity.} A more precise complexity landscape can be obtained by taking the perspective of parameterized complexity: we say that a problem is \emph{fixed-parameter tractable} (\FPT, or \emph{in \FPT}, for short) parameterized by a \emph{parameter} $k$ if there is an algorithm solving any instance $I$ in time $f(k)\poly(|I|)$, for some computable function $f$.
On the other hand, showing that a problem is \emph{\Wh{1}} means it is unlikely to have such an algorithm, and the best one might hope for is a complexity of the form $|I|^{f(k)}$; we then say that a problem is \emph{in \XP} (or that it \emph{has an \XP algorithm}); see the textbook~\cite{CyganFKLMPPS15}.

The hard instance $I$ from Theorem~\ref{thm:QCmaxNPhard} is encoded by a job sizes matrix $\vep$, a job multiplicities vector $\ven$, and a machine speeds vector $\ves$ which all contain long numbers, i.e., entries with encoding length $\Omega(|I|)$.
What happens when some of $\vep$, $\ven$, and $\ves$ are restricted to numbers bounded by $\poly(|I|)$, or, equivalently, if they are encoded in unary?

A note of caution: since we allow speeds to be rational, and the encoding length of a fraction $p/q$ is $\ceil{\log_2 p} + \ceil{\log_2 q}$,  a $Q||C_{\max}$ instance with $\ves$ of polynomial length might translate to an $R||C_{\max}$ instance with $\vep$ of exponential length.
This is because for $\vep$ to be integer, one needs to scale it up by the least common multiple of the denominators in $\ves$, which may be exponential in $m$.
Thus, with respect to the magnitude of $\ven$ and $\vep$, $R|HM|C_{\max}$ can \emph{not} be treated as a generalization of $Q|HM|C_{\max}$.
This is why in the following we deal with both problems and not just the seemingly more or less general one.
For $Q|HM|C_{\max}$, we denote by $p_{\max}$ the largest job size \emph{before scaling}, i.e., if $\vep = \vep' \cdot (1/\ves)^{\intercal}$, then $p_{\max} = \|\vep'\|_\infty$.

Having $\ven$ polynomially bounded is equivalent to giving each job explicitly; note that in this setting $R|HM|C_{\max}$ strictly generalizes $Q|HM|C_{\max}$.
A simple DP handles this case:
\begin{theorem} \label{thm:dp}
	$\{R,Q\}|HM|C_{\max}$ and $\{R,Q\}||C_{\max}$ can be solved in time $m \cdot n^{\Oh(k)}$, hence $\{R,Q\}||C_{\max}$ is in \XP parameterized by $k$.
\end{theorem}
A similar situation occurs if $\ven$ is allowed to be large, but $\vep$ is polynomially bounded, although the use of certain integer programming tools~\cite{EisenbrandEtAl2019} is required:
\begin{theorem}\label{thm:proximity}
	$\{R,Q\}|HM|C_{\max}$ can be solved in time $p_{\max}^{\Oh(k^2)} m \log m \log^2 n$, hence \\
	$\{R,Q\}|HM|C_{\max}$ are in \XP parameterized by $k$ if $p_{\max}$ is given in unary.
\end{theorem}
Our main result is that an \FPT algorithm for $Q|HM|C_{\max}$ is unlikely to exist even when $\ven$, $\vep$, and $\ves$ are encoded in unary, and for $R|HM|C_{\max}$ even when the rank of $\vep$ is $2$:
\begin{theorem}\label{thm:xcmaxwh1}
	$X||C_{\max}$ is \Wh{1} parameterized by the number of job types with
	\begin{enumerate*}
\item $X=Q$ and $\ven$, $\vep$, and $\ves$ given in unary.
\item $X=R$ and $\ven$ and $\vep$ given in unary and $\rank(\vep) = 2$.
	\end{enumerate*}
\end{theorem}
We use a result of Jansen et al.~\cite{JansenKMS:2013} as the basis of our hardness reduction.
They show that \textsc{Bin Packing} is \Wh{1} parameterized by the number of bins even if the items are given in unary.
In the context of scheduling, this means that $P||C_{\max}$ is \Wh{1} parameterized by the number of machines already when $p_{\max}$ is polynomially bounded.
However, it is non-obvious how to ``transpose'' the parameters, that is, how to go from many job types and few machines to few job types and many machines which differ as little as possible (i.e., only by their speeds, or only in low-rank way).
We first show \Whness{1} of \balancedbinpacking, where we additionally require that the number of items in each bin is identical, parameterized by the number of bins, even for tight instances in which each bin has to be full.
Using this additional property, we are able to construct an $R|HM|C_{\max}$\jzcom{I think we should change this to $Q|HM|C_{\max}$ and then introduce the blocker job types for $R|HM|C_{\max}$} instance of makespan $T$ in which optimal solutions are in bijection with optimal packings of the encoded \balancedbinpacking instance.
Our $R|HM|C_{\max}$ instance uses one job type to ``block out'' a large part of a machine's capacity so that its remaining capacity depends on the item the machine represents, and all other job types have sizes independent of which machine they run on.
Since the capacity of a machine exactly corresponds to its speed, omitting those ``blocker'' jobs and setting the machine speeds gives a hard instance for $Q|HM|C_{\max}$.

Let us go back to $P|HM|C_{\max}$.
As mentioned previously, Goemans and Rothvoss showed that if the largest job size $p_{\max}$ is polynomially bounded, the problem is \FPT because $(\log p_{\max})^{f(k)} \poly\log n \leq g(k) \cdot p^{o(1)}_{\max} \poly\log n$~\cite[Exercise 3.18]{CyganFKLMPPS15}.
We answer the remaining question whether the problem is in \FPT also when all jobs are given explicitly:
\begin{theorem} \label{thm:pcmax}
	$P||C_{\max}$ is \FPT parameterized by $k$.
\end{theorem}
This result partially answers~\cite[Question 5]{MnichB18}, which asks for an \FPT algorithm for $P|HM|C_{\max}$.
Obtaining this answer turns out to be surprisingly easy: we reduce the job sizes by a famous algorithm of Frank and Tardos~\cite{FT} and then apply the algorithm of Goemans and Rothvoss~\cite{GoemansRothvoss2014}, which is possible precisely when $n$ is sufficiently small.
This extends our understanding of the complexity of $P|HM|C_{\max}$: the problem is \FPT if either the largest job or the number of jobs are not too large.
Hence, the remaining (and major) open problem is the complexity of $P|HM|C_{\max}$ parameterized by $k$, without any further assumptions on the magnitude of $p_{\max}$ or $n$.
In light of this, our result that already $Q|HM|C_{\max}$ is \NPh when $p_{\max}$ and $n$ are large, and \Wh{1} if both are polynomially bounded, may be interpreted as indication that the magnitude of $n$ and $p_{\max}$ plays a surprisingly important role, and that $P|HM|C_{\max}$ may in fact \emph{not} be \FPT parameterized by~$k$.

\begin{table}[]
	\def\arraystretch{1.1}
	\setlength{\tabcolsep}{0.5em}
	
	\makeatletter
	\newcommand{\thickhline}{%
		\noalign {\ifnum 0=`}\fi \hrule height 1.2pt
		\futurelet \reserved@a \@xhline
	}
	\newcolumntype{"}{@{\hskip\tabcolsep\vrule width 1.2pt\hskip\tabcolsep}}
	
	\begin{tabular}{l"c|cc|cc"c|c|c}
		& $P||\dots$ & \multicolumn{2}{c|}{$Q||\dots$} & \multicolumn{2}{c"}{$R||\dots$} & $P|HM|\dots$ & $Q|HM|\dots$ & $R|HM|\dots$ \\ \thickhline
		$C_{\max}$ & \FPT & \multirow{9}{*}{\rotatebox[origin=c]{270}{\XP (Theorem~\ref{thm:dp})}} & \Wh{1} & \multirow{9}{*}{\rotatebox[origin=c]{270}{\XP (Theorem~\ref{thm:dp})}} & \Wh{1} & poly.\ time & \NPh & \NPh \\
		& (Thm.~\ref{thm:pcmax}) & & (Thm.~\ref{thm:xcmaxwh1}) & & (Thm.~\ref{thm:xcmaxwh1}) & for const.\ $k$ & for $k \ge 6$ & for $k \ge 4$ \\
		& & & & & & (\cite{GoemansRothvoss2014}) & (Thm.~\ref{thm:QCmaxNPhard}) & (Thm.~\ref{thm:RCmaxNPhard}) \\ \cline{1-2}\cline{4-4}\cline{6-9}
		$\ell_2$ & ? & & \Wh{1} & & \Wh{1} & ? & \NPh & \NPh \\
		& & & (Cor.~\ref{cor:Xl2W1hard}) & & (Cor.~\ref{cor:Xl2W1hard}) & & for $k \ge 6$ & for $k \ge 7$ \\
		& & & & & & & (Cor.~\ref{cor:Xl2NPhard}) & (Cor.~\ref{cor:Xl2NPhard}) \\ \cline{1-2}\cline{4-4}\cline{6-9}
		$\sum w_j C_j$ & ? & & ? & & \Wh{1} & ? & ? & \NPh \\
		& & & & & (Cor.~\ref{cor:RweightedSumW1hard}) & & & for $k \ge 7$ \\
		& & & & & & & & (Cor.~\ref{cor:RweightedSumNPhard}) \\
	\end{tabular}
	
	\bigskip
	
	\caption{Overview of the computational hardness of $\{P, Q, R\}|\{\_, HM\} | \{C_{\max}, \ell_2, \sum w_j C_j\}$ relative to the number of job types~$k$.}
	\label{tab:overview}
\end{table}

\vspace{-1em}
\subparagraph*{Other Objectives.}
Besides minimum makespan, two important scheduling objectives are minimization of the sum of weighted completion times, denoted $\sum w_j C_j$, and the minimization of the $\ell_2$-norm of the load vector.
We show that our algorithms and hardness results (almost always) translate to these objectives as well.
Let us now introduce them formally.

The load $L_i$ of a machine $M_i$ is the total size of jobs assigned to it.
In $R|HM|\ell_2$, the task is to find a schedule minimizing $\|(L_1, \dots, L_m)\|_2 = \sqrt{\sum_{i=1}^m L_i^2}$.
Note that this is isotonic (order preserving) to the function $\sum_{i=1}^m L_i^2$, and because this leads to simpler proofs, we instead study the problem $R|HM|\ell_2^2$.
The completion time of a job, denoted $C_j$, is the time it finishes its execution in a schedule.
In the $R|HM|\sum w_j C_j$ problem, each job is additionally given a weight $w_j$, and the task is to minimize $\sum w_j C_j$.

We show that the hard instance for $R|HM|C_{\max}$ is also hard for $\ell_2$, and with the right choice of weights is also hard for $\sum w_j C_j$.
We also obtain hardness of $Q|HM|\ell_2$ by a different and more involved choice of speeds, but the case of $Q|HM|\sum w_j C_j$ remains open so far.
To extend the $C_{\max}$ reduction to other objectives, we use the ``tightness'' of our hardness instance to show that any ``non-tight'' schedule must increase the $\ell_2$ norm of the load vector by at least some amount.
This is not enough for $R|HM|\sum w_j C_j$ because the value $\sum w_j C_j$ is proportional to the load vector plus other terms, and we need to bound those remaining terms (Lemma~\ref{lem:gapsumwjCj}) in order to transfer the argument from $\ell_2$ to $\sum w_j C_j$.
We point out that the these hardness results are delicate and non-trivial even if at first sight they may appear as ``just'' modifying the hard instance of $Q|HM|C_{\max}$.

\ifproc Proofs of statements marked with ``\appmark'' are available in the full version. \fi
We give an overview of our results in Table~\ref{tab:overview}.

\section{Preliminaries}
We consider zero a natural number, i.e., $0 \in \N$.
We write vectors in boldface (e.g., $\vex, \vey$) and their entries in normal font (e.g., the $i$-th entry of a vector~$\vex$ is~$x_i$).
If it is clear from context that $\vex^{\intercal} \vey$ is a dot-product of $\vex$ and $\vey$, we just write $\vex \vey$~\cite{ConfortiCZ:2014}.
We use $\log \df \log_2$, i.e., all our logarithms are base~$2$.
For $n,m \in \N$, we write $[n,m] = \{n, n+1, \dots, m\}$ and $[n] = [1,n]$.
%

\jzcom{add also $Q|HM|C_{\max}$}%

\prob{\textsc{Makespan Minimization on Unrelated Machines} ($R|HM|C_{\max}$)}
{$n$ jobs of $k$ types, job multiplicities $n_1, \dots, n_k$, i.e., $n_1 + \cdots + n_k = n$ and $n_j$ is the number of jobs of type $j$, $m$ unrelated machines, for each $i \in [m]$ a job sizes vector $\vep^i = (p^i_1, \dots, p^i_k) \in (\N \cup \{+\infty\})^{k \cdot m}$ where $p^i_j$ is the processing time of a job of type $j$ on a machine $M_i$, a number $T$.}
{An assignment of jobs to machines and non-overlapping (with respect to each machine) time slots such that every machine finishes by time $T$.}

Notice that our definition uses a \emph{high-multiplicity} encoding of the input, that is, jobs are not given explicitly, one by one, but ``in bulk'' by a vector of multiplicities.
Because this allows compactly encoding instances which would otherwise be of exponential size, the two problems actually have different complexities and deserve a notational distinction: we denote by $R||C_{\max}$ the problem where jobs are given explicitly, and by $R|HM|C_{\max}$ the problem defined above; see also the discussion in~\cite{KnopK2020}.

Recall that in $R|HM|\ell_2$, the task is to minimize $\|(L_1, \dots, L_m)\|_2$, where $L_i$ is the sum of sizes of jobs assigned to machine $M_i$ for $i \in [m]$.
In $R|HM|\sum w_j C_j$, each job~$j$ has a weight~$w_j$, and a schedule determines a job's completion time $C_j$.
The task is then to minimize $\sum w_j C_j$.

The job sizes matrix $\vep \in \R_+^{k \times m}$ has rank $r$ if it can be written as a product of matrices $C \in \R^{k \times r}$ and $D \in \R^{r \times m}$.
For example, in $Q||C_{\max}$, each machine has a speed $s_i \in \R_+$, and $\vep^i = \vep' / s_i$ for some $\vep' \in \N^k$, so $\vep = \vep' (1/\ves)^{\intercal}$, where $\ves = (s_1, \dots, s_m)$, hence $\vep$ has rank $1$.

In the \emph{identical machines} model, $\vep^i = \vep$ for all $i \in [m]$, and we denote it $P||C_{\max}$.
Its decision variant $P||C_{\max}$ is equivalent to \binpacking:

\prob{\binpacking}
{$n$ items of sizes $a_1, \dots, a_n$, $k$ bins, each with capacity $B$.}
{An assignment of items to bins such that the total size of items in each bin is $\leq B$.}

\unarybinpacking is \binpacking where all $a_1, \dots, a_n$ are encoded in unary, or, equivalently, $a_{\max} = \max_i a_i$ is bounded polynomially in $n$.
\balancedbinpacking is \binpacking with the additional requirement on the solution that the number of items assigned to each bin is the same, hence $n/k$; note that $n$ has to be divisible by $k$ for any instance to be feasible.
An instance of \binpacking is \emph{tight} if the total size of items $\sum_i a_i$ is equal to $k \cdot B$, which means that if an instance has a packing, then each bin is used fully.

\section{Algorithms}
We wish to highlight the geometric structure of $R|HM|C_{\max}$ by formulating it as an ILP and making several observations about it.
We have a variable $x_j^i$ for each job type $j \in [k]$ and machine~$M_i$ (with $i \in [m]$) specifying how many jobs of type $j$ are scheduled to run on machine $M_i$.
There are two types of constraints, besides the obvious bounds $\vezero \leq \vex^i \leq \ven$ for each $i \in [m]$.
The first enforces that each job is scheduled somewhere, and the second assures that the sum of job sizes on each machine is at most $T$, meaning each machine finishes by time $T$:
\begin{align}
\sum_{i = 1}^{m} x^i_j &= n_j & \forall j \in [k]& \label{eq:1}\\
\sum_{j = 1}^{k} x^i_j p^i_j &\leq T & \forall i \in [m]& \enspace . \label{eq:2}
\end{align}
Knop and Koutecký~\cite{KnopK:2017} show that this ILP has \emph{$N$-fold} format, i.e., it has the general form:
\begin{align*}
\min f(\vex): \,
E^{(N)}\vex = \veb, \, \vel \leq \vex \leq \veu,\ \vex \in \Z^{Nt}, \text{ with }
E^{(N)} = \left(
\begin{array}{cccc}
E^1_1    & E^2_1    & \cdots & E^N_1    \\
E^1_2    & 0      & \cdots & 0      \\
0    & E^2_2    & \cdots & 0      \\
\vdots & \vdots & \ddots & \vdots \\
0    & 0      & \cdots & E^N_2    \\
\end{array}
\right) \enspace .
\end{align*}
Here, $r,s,t,N \in \N$, $E^{(N)}$ is an $(r+Ns)\times Nt$-matrix, $E^i_1 \in \Z^{r \times t}$ and $E^i_2 \in \Z^{s \times t}$ for all $i \in [N]$, are integer matrices, and $f$ is some separable convex function.
Specifically for $R||C_{\max}$, $f \equiv 0$, the matrices corresponding to equations~\eqref{eq:1}--\eqref{eq:2} are $E_1^i = I$ and $E_2^i = \vep^i$, for each $i \in [m]$, $\veb = (\ven, T, \dots, T)$ is an $r+Ns = (k+m)$-dimensional vector, and $\vel=\vezero$ and $\veu = (\ven, \ven, \dots, \ven)$ are $Nt = (mk)$-dimensional vectors.
We note that $N$-fold IP formulations are also known for $R|HM|\{\ell_2, \sum w_j C_j\}$~\cite{KnopK2020,KnopK:2017}.

\subsection{Large Lengths, Polynomial Multiplicities}
A simple dynamic programming algorithm gives:
\ifproc
\begin{reptheorem}{thm:dp}[\appmark]
\else
\begin{reptheorem}{thm:dp}
\fi
$\{R,Q\}|HM|\{C_{\max},\ell_2,\sum w_j C_j\}$ can be solved in time $m \cdot n^{\Oh(k)}$, hence $\{R,Q\}||\{C_{\max},\ell_2,\sum w_j C_j\}$ are in \XP parameterized by $k$.
\end{reptheorem}
\ifproc\else
\begin{proof}
We will describe a simple dynamic programming (DP) algorithm.
Call a vector $\vex^i \in \N^k$ satisfying the constraint~\eqref{eq:2}, i.e., $\vep^i \vex^i \leq T$, a \emph{configuration of machine~$M_i$}.
We will construct a DP table $D$ indexed by $k$-dimensional integer vectors upper bounded by $\ven$, and $i \in [m]$, and each value of the table is a 0/1 bit.
The intended meaning is that, for $i \in [m]$ and $\ven' \leq \ven$, $D[i, \ven'] = 1$ iff the subinstance consisting of jobs $\ven'$ and the first $i$ machines is feasible.
Initialize $D$ to be all-zero, and set $D[0, \vezero] = 1$.
Then, consecutively for $i=1, \dots, m$, and for each $\vezero \leq \ven' \leq \ven$, set $D[i, \ven'] = 1$ if $D[i-1, \ven'- \vex^i] = 1$ and $\vex^i$ is a configuration of machine~$M_i$.
In other words, for each $i=1,\dots, m$, construct the set $\CC^i$ of configurations of machine~$M_i$, and then, for each $\ven'$ with $D[i-1,\ven'] = 1$, set $D[i, \ven'+\vex^i]=1$ for each $\vex^i \in \CC^i$ if $\ven' + \vex^i \leq \ven$.
Finally, the instance is feasible if $D[m,\ven]=1$.
In each iteration, we go over all $\ven' \leq \ven$, of which there is at most $n^k$ many, and for each of them, we try to add each element of $\CC^i$, of which there is also at most $n^k$ many.
In total, the algorithm makes $m \cdot n^k \cdot n^k = m \cdot n^{2k}$ steps.

The adaptation of this DP to $\ell_2$ and $\sum w_j C_j$ is straightforward.
Say that a configuration is any vector $\vex^i \leq \ven$.
The value of a configuration $\vex^i$ on machine~$M_i$ is $f^i(\vex^i) = (\vep^i \vex^i)^2$ for $\ell^2_2$.
For $\sum w_j C_j$, it has been shown~\cite{KnopK:2017} that the contribution of a machine~$M_i$ scheduling jobs $\vex^i$ is a quadratic convex function $f^i$ in terms of $\vex^i$.
Then, $D[i, \ven'] = \min_{\vex^i \leq \ven-\ven'} f^i(\vex^i) + D[i-1, \ven' - \vex^i]$.
\end{proof}
\fi

Theorem~\ref{thm:dp} (with a worse complexity bound) can be also shown in a somewhat roundabout way by manipulating the ILP formulation \eqref{eq:1}--\eqref{eq:2}.
This approach will eventually give us the result that $P||C_{\max}$ is \FPT parameterized by $k$.
We need the following result:
\begin{proposition}[{Frank and Tardos~\cite{FT}}]\label{prop:FT}
	Given a rational vector $\vew \in \Q^{d}$ and an integer $M$, there is a strongly polynomial algorithm which finds a $\bar{\vew} \in \Z^d$ such that for every integer point $\vex\in [-M,M]^d$, we have $\vew \vex \geq 0 \Leftrightarrow \bar{\vew} \vex \geq 0$ and $\|\bar{\vew}\|_\infty \leq 2^{\Oh(d^3)} M^{\Oh(d^2)}$.
\end{proposition}

\begin{lemma}\label{lem:nfold-reduce-coeff}
It is possible to compute in strongly-polynomial time for each $i \in [m]$ a vector $\bar{\vep}^i \in \N^k$ and an integer $\bar{T}^i \in \N$ such that replacing constraint~\eqref{eq:2} with $\bar{\vep}^i \vex^i \leq \bar{T}^i$ does not change the set of feasible integer solutions, and $\|\bar{\vep}^i, \bar{T}^i\|_\infty \leq 2^{\Oh(k^3)} n^{\Oh(k^2)}$
\end{lemma}
\begin{proof}
Fix some $i \in [m]$ and consider the inequality~\eqref{eq:2}, which is $\vep^i \vex^i \leq T$.
Applying Proposition~\ref{prop:FT} to $(\vep^i,T)$ and $M=n$ gives a vector $(\bar{\vep}^i, \bar{T}^i)$ such that for all $\vezero \leq \vex^i \leq \ven$,
\[
(\vep^i,T) (\vex^i,-1) \leq 0 \Leftrightarrow (\bar{\vep}^i,\bar{T}^i) (\vex^i,-1) \leq 0,
\]
which means that replacing $\vep^i \vex^i \leq T$ by $\bar{\vep}^i \vex^i \leq \bar{T}$ in~\eqref{eq:2} does not change the set of feasible solutions, and the bound on $\|\bar{\vep}^i, \bar{T}\|_\infty$ follows immediately from Proposition~\ref{prop:FT}.
\end{proof}
We will use the fact that $N$-fold IP can be solved efficiently:
\begin{proposition}[\cite{jansen2018near,cslovjecsek2020n,EisenbrandEtAl2019}] \label{prop:nfold}
A feasibility instance of $N$-fold IP can be solved in time $(\|E^{(N)}\|_\infty rs)^{\Oh(r^2s + s^2)} Nt \log Nt \log^2 \|\veu-\vel\|_\infty$.
\end{proposition}

\begin{proof}[{Alternative proof of Theorem~\ref{thm:dp}} for $C_{\max}$]
By Lemma~\ref{lem:nfold-reduce-coeff}, we can reduce $\|E^{(N)}\|_\infty$ down to $2^{\Oh(k^3)} n^{\Oh(k^2)}$.
Since $r=k$, $t=k$, $s=1$, $N=m$, and $\|\veu-\vel\|_\infty \leq n$, applying Proposition~\ref{prop:nfold} to such a reduced instance gives an $n^{\Oh(k^5)} m \log m \log^2 n$ algorithm.
Dealing with $\ell_2$ and $\sum w_j C_j$ is analogous, see Lemma~\ref{lem:nfoldreduction}.
\end{proof}
While this is worse than the DP above, notice that this approach also gives:
\begin{reptheorem}{thm:pcmax}
$P||C_{\max}$ is \FPT parameterized by $k$.
\end{reptheorem}
\begin{proof}
Apply Lemma~\ref{lem:nfold-reduce-coeff} to a given $P||C_{\max}$ instance, which gives a new job-sizes vector $\bar{\vep} \in \N^k$ and a new time bound $\bar{T} \in \N$.
Goemans and Rothvoss~\cite{GoemansRothvoss2014} have shown that $P||C_{\max}$ with $k$ job types can be solved in time $(\log p_{\max})^{2^{\Oh(k)}} \poly\log n$.
Plugging in $p_{\max} \leq 2^{\Oh(k^3)} n^{\Oh(k^2)}$ gives $\log p_{\max} \leq \log 2^{\Oh(k^3)} n^{\Oh(k^2)} = k^3 + k^2 \log n$.
Hence, the algorithm runs in time $(k^3 \log n)^{2^{\Oh(k)}} = (k^3)^{2^{\Oh(k)}} \cdot (\log n)^{2^{\Oh(k)}}$.
To verify that this is indeed an \FPT runtime
(i.e., $f(k)\poly(n)$ for some computable $f$),
we use a simple observation~\cite[Exercise 3.18]{CyganFKLMPPS15} that $(\log \alpha)^\beta \leq 2^{\beta^2/2} \alpha^{o(1)}$.
Taking $\alpha = n$ and $\beta = 2^{\Oh(k)}$ gives $(\log n)^{2^{\Oh(k)}} \leq 2^{2^{\Oh(k)}} n^{o(1)}$ and we are done.
\end{proof}
\begin{remark}
The algorithm of~\cite{GoemansRothvoss2014} shows that $P|HM|C_{\max}$ is \FPT in $k$ if $p_{\max}$ is given in unary.
To the best of our knowledge, it has not been observed before that $P|HM|C_{\max}$ is \FPT in $k$ if $n$ is polynomially bounded by the input length, i.e., that $P||C_{\max}$ is \FPT in $k$.
Thus, Theorem~\ref{thm:pcmax} shows that the remaining (and indeed hard) open problem is the complexity of $P|HM|C_{\max}$ for instances where both $\vep$ and $\ven$ contain large numbers.
\end{remark}

A straightforward adaptation of the proof of Lemma~\ref{lem:nfold-reduce-coeff} where we reduce each row of the constraint $E_2^i \vex^i = \veb^i$ separately gives the following more general statement:
\begin{lemma}\label{lem:nfoldreduction}
Given an $N$-fold IP instance and $M \in \N$, one can in strongly-polynomial time compute $\bar{E}_2^i$ and $\bar{\veb}^i$, for each $i \in [N]$, such that if $\|\veu - \vel\|_\infty \leq 2M$, then
\[
\{\vex \in \Z^{Nt} \mid E^{(N)}\vex = \veb, \, \vel \leq \vex \leq \veu\} =
\{\vex \in \Z^{Nt} \mid \bar{E}^{(N)}\vex = \bar{\veb}, \, \vel \leq \vex \leq \veu\},
\]
where $\bar{E}^{(N)}$ is obtained from $E^{(N)}$ by replacing $E_2^i$ with $\bar{E}_2^i$ and $\bar{\veb}$ is obtained from $\veb$ by replacing $\veb^i$ with $\bar{\veb}^i$, for each $i \in [N]$, and $\|\bar{E}_2^i, \bar{\veb}^i\|_\infty \leq 2^{\Oh(t^3)} M^{\Oh(t^2)}$. \qed
\end{lemma}

\subsection{Polynomial Lengths, Large Multiplicities}
How to deal with instances whose jobs have polynomially bounded sizes, but come in large multiplicities?
Actually, the fact that $R|HM|C_{\max}$ belongs to \XP parameterized by $k$ if $p_{\max}$ is polynomially bounded follows by solving the $N$-fold IP~\eqref{eq:1}--\eqref{eq:2} using Proposition~\ref{prop:nfold}:
\begin{reptheorem}{thm:proximity}
$\{R,Q\}|HM|\{C_{\max}, \ell_2, \sum w_j C_j\}$ can be solved in time $p_{\max}^{\Oh(k^2)} m \log m \log^2 n$.
\end{reptheorem} 
To obtain a result like this one can first solve the LP relaxation of~\eqref{eq:1}--\eqref{eq:2}, and then use a ``proximity theorem'' to show that some integral optimum is at distance at most $p_{\max}^{\Oh(k)}\cdot m$~\cite[Theorem 59]{EisenbrandEtAl2019} from any optimum of the LP relaxation.
This yields an $\{R,Q\}|HM|\{C_{\max}, \ell_2, \sum w_j C_j\}$ instance where roughly $p_{\max}^{k}\cdot m$ jobs are left to be scheduled and which can be solved using Theorem~\ref{thm:dp}.
To adapt the model~\eqref{eq:1}--\eqref{eq:2} for uniformly related machines, one has a single vector $\vep \in \N^\tau$ of ``unscaled'' processing times, and the right hand side of constraint~\eqref{eq:2} becomes $\floor{T\cdot s_i}$ for a machine of speed $s_i$.
For $\ell_2$, the objective $f$ of the $N$-fold formulation becomes $f(\vex) = \sum_{i=1}^m (\vep^i \vex^i)^2$ which is almost separable convex (one needs to add an auxiliary variable $z^i$ and a constraint $z^i = \vep^i \vex^i$ to express it as separable).
For $\sum w_j C_j$, the modification is analogous but slightly more complicated; the approach is identical to the one described by Knop and Koutecký~\cite{KnopK:2017}.

It is an open problem whether the $p_{\max}^{\Oh(k^2)}$ parameter dependence can be improved: even in the setting with short jobs where $p_{\max} \leq k$, the best algorithm for $Q|HM|C_{\max}$ has a dependence of $k^{k^2}$~\cite{KnopKLMO:2019,KnopK:2017}.


\section{Hardness}


\subsection{Reducing \binpacking to \balancedbinpacking}
\begin{lemma}
	\label{lem:reductionBPtoBBP}
	\binpacking reduces to \balancedbinpacking such that \\
	\begin{enumerate*}
		\item $a'_{\max} = a_{\max} + 1$,
		\item $B' = B+n$,
		\item $k'= k$,
		\item $n' = nk$, and
		\item tightness is preserved,
	\end{enumerate*}
	where $n', k', B', a'_{\max}$ are the parameters of the new \balancedbinpacking~instance.
\end{lemma}

\begin{proof}
	Given an instance of \binpacking,	
	we obtain an instance of \balancedbinpacking by increasing the size of each item by~$1$, setting the new bin capacity to be $B' = B+n$, and adding $n (k-1)$ new items of size~$1$.
	Observe that all items of size~$1$ are ``new'' items.
	It is also clear that $a'_{\max} = a_{\max}+1$.
	
	To show that we preserve feasibility of instances, take any solution of the \binpacking instance and add new items of size zero such that each bin contains precisely $n$ items.
	Now if we increase the size of each item by $1$ (including the new items of size zero) and the size of each bin by~$n$,
	we have obtained a feasible instance of the newly constructed \balancedbinpacking instance.
	
	For the other direction,
	assume for the sake of contradiction that the \balancedbinpacking instance has a solution, but the original \binpacking instance does not.
	Consider a solution of \balancedbinpacking, subtract $1$ from the size of each item and $n$ from the capacity of each bin---note that there are $n$ items per bin---and remove items of size zero.
	This is a solution to the instance of \binpacking---a contradiction.
	
	Regarding tightness, note that the sum of item sizes has increased by exactly $nk$ because we have increased the size by $1$ for $n$ ``old'' items, and added $n(k-1)$ ``new'' items of size 1.
	Hence, if the total size of items of the original instance was $kB$, it became $kB + nk = k(B+n)$, and since $B'=B+n$ is the new bin capacity, the \balancedbinpacking instance is tight iff the \binpacking instance was.
\end{proof}

\begin{corollary}
	\balancedbinpacking is \NPh, even for tight instances.
\end{corollary}

\begin{corollary}
	\label{cor:unarybalancedbinpacking-is-W1-hard}
	\unarybalancedbinpacking is \Wh{1} parameterized by the number of bins, even for tight instances.
\end{corollary}

\subsection{Hardness of $Q||C_{\max}$ and $R||C_{\max}$}
\label{sec:hardnessX||Cmax}

Let us describe our hard instance $I$.
Given a tight instance of \balancedbinpacking with $k$ bins of capacity~$B$ and $m$~items,
all items sum up to $\sum_{i \in [m]} a_i = k \cdot B =: A$.
We construct a $Q|HM|C_{\max}$ instance with $m$ machines and $3 k$ job types.

The high level idea is as follows.
We use machine~$M_i$ to encode the assignment of item $a_i$ to a bin, so we have $m$ machines.
We have job types $\alpha^1_j, \alpha^0_j$ (we will refer to both of them as $\alpha^{\times}_j$), and~$\beta_j$ for $j \in [k]$; we refer to a job of type $\alpha^\times_j$ for any $j$ as a job of type $\alpha$ or an $\alpha$-type job, and similarly for $\beta$.
For the sake of simplicity, we sometimes do not distinguish between a job and a job type, e.g., by executing $\alpha^\times_j$ we mean executing a job of type $\alpha^\times_j$.

Our goal is to ensure that a specific schedule, which we call henceforth \emph{perfect}, is optimal.
In a perfect schedule, $M_i$ gets precisely $a_i$ times a job of type~$\alpha^1_j$, $A - a_i$ times a job of type~$\alpha^0_j$ and once a job of type~$\beta_j$ for some $j \in [k]$.
There is no other job on~$M_i$.
This corresponds to putting $a_i$ to the $j$-th bin.
Hence, for each $j \in [k]$, there are $m/k$ machines\footnote{Which is an integer by the fact that any \balancedbinpacking instance must have a number of items divisible by $k$ in order to be feasible.} where only jobs of types~$\alpha^1_j$, $\alpha^0_j$ and $\beta_j$ appear together
and they represent a packing of the corresponding items to the~$j$-th bin.

Let us specify the parameters of $I$.
The target makespan is $T = 3 k A^3$; note that we will show that the feasible schedules are precisely the perfect schedules and they have the property that each machine finishes \emph{exactly} at time~$T$.
Jobs of type~$\beta$ are by far the largest on all machines.
We set, for $j \in [k]$,
\begin{equation*}
p_{\alpha^1_j} = k A^2 + A (k - j) + 1 \, , \qquad
p_{\alpha^0_j} = k A^2 + A (k - j) \, , \qquad
p_{\beta_j} = 2 k A^3 - A^2 (k-j) \, ;
\end{equation*}
note that as $j$ increases, so does $p_{\beta_j}$.
Complementary to $p_{\beta_j}$, as $j$ increases, $p_{\alpha^{\times}_j}$ decreases.
To show hardness of~$Q||C_{\max}$, we give each machine~$M_i$ a specific speed depending on~$a_i$.
The \emph{unscaled load of a machine $M_i$}, denoted $\bar{L}_i$, is the sum of sizes of jobs assigned to $M_i$ before speed scaling.
In a perfect schedule, it is
\begin{align}
\bar{L}_i^* &= a_i (k A^2 + A (k - j) + 1) + (A - a_i) (k A^2 + A (k - j)) + 2 k A^3 - A^2 (k-j) \nonumber \\
&= A (k A^2 + A (k - j)) + a_i + 2 k A^3 - A^2 (k-j)
= 3 k A^3 + a_i = T + a_i \, .
\label{eq:perfectSchedule}
\end{align}
The machine speed~$s_i$ of machine~$M_i$ is
\begin{equation*}
s_i = \frac{T + a_i}{T} = \frac{3 k A^3 + a_i}{3 k A^3} \, .
\end{equation*}
Observe that in a perfect schedule each machine~$M_i$ finishes exactly by time
\begin{equation}
\label{eq:machineLoadPerfectSchedule}
\frac{\bar{L}_i^*}{s_i} = \frac{T + a_i}{\frac{T + a_i}{T}} = T = 3 k A^3 \, .
\end{equation}
The sizes of jobs of type $\alpha^1_j$ and $\alpha^0_j$ are almost identical, except jobs of type~$\alpha^1_j$ are slightly longer.
For each $j \in [k]$, we have job multiplicities
\begin{equation*}
n_{\alpha^1_j} = \frac{A}{k} = B, \qquad n_{\alpha^0_j} = \frac{A m}{k} - B =\frac{(m - 1) A}{k}, \qquad n_{\beta_j} = \frac{m}{k} \enspace .
\end{equation*}

\begin{lemma}
	\label{lem:reductionBBPtoQCmax}
	\balancedbinpacking with tight instances reduces to $Q|HM|C_{\max}$ such that
	\begin{enumerate*}
		\item the number of machines equals the number of items,
		\item the number of job types equals $3k$, where $k$ is the number of bins,
		\item the job sizes and job multiplicities are bounded by $\Oh(A^4)$, where $A$ is the sum of all items of the input instance,\label{lem:reductionBBPtoQCmax-item:numbersPolynomial}
		\item the machine speeds are rational numbers with numerator and denominator in~$\Oh(A^4)$, and
		\item the feasible schedules are precisely perfect schedules, in which all machines finish exactly at time~$T = 3 k A^3$.
		\label{lem:reductionBBPtoQCmax-item:allMachinesReachMakespan}
	\end{enumerate*}
\end{lemma}
\begin{proof}
	Clearly, all involved numbers are in $\Oh(A^4)$ (w.l.o.g.\ we assume $k, m \in \Oh(A)$).
	The other parameters are clear from the description of the hard instance~$I$ above.
	It remains to prove the correctness of our reduction.
	On the one hand, if there is a solution~$\mathcal{S}$ of the corresponding instance of \balancedbinpacking,
	we construct a (feasible) perfect schedule for~$I$ as follows.
	If, in~$\mathcal{S}$, $a_i$ is assigned to the $j$-th bin,
	to machine~$M_i$ we assign $a_i$ jobs of type~$\alpha^1_j$, $A - a_i$ jobs of type~$\alpha^0_j$, and one job of type~$\beta_j$.
	According to equations~\eqref{eq:perfectSchedule} and~\eqref{eq:machineLoadPerfectSchedule}, this assignment has makespan~$T$ and, clearly, all jobs are assigned to some machine.
	
	On the other hand, assume that~$I$ is feasible, meaning there is an assignment of jobs to machines not exceeding the target makespan~$T$.
	Let us analyze the structure of such a schedule~$\sigma$.
	First we observe that instead of considering for a machine~$M_i$ the makespan~$T$, which is the sum of jobs lengths divided by its speed~$s_i$, we can equivalently consider $T \cdot s_i = T + a_i$ as its capacity---this is the sum of (unscaled) jobs lengths it can process.
	Per machine, there is exactly one job of type~$\beta_j$ for some $j \in [k]$, since we can execute at most one $\beta$-type job on each machine and we have to place $m$ such jobs onto $m$ machines.
	So each machine is in one set~$\mathcal{M}_j$, where $\mathcal{M}_j$ is a set of $m/k$ machines that process a job of type~$\beta_j$.
	Having scheduled a job of type~$\beta_j$ to a machine, we can execute on this machine at most $A$ jobs of type $\alpha^{\times}_{j'}$ for any $j'$.
	In particular, observe that even on a machine that executes $\beta_1$, which is the smallest of the $\beta$-type jobs, we cannot add $A+1$ jobs of type $\alpha^0_{k}$, which is the smallest of the $\alpha^{\times}_j$ job types, without exceeding $T + \max_i a_i$.
	
	For each $j \in [k]$, there are $Am/k$ jobs of type~$\alpha^{\times}_{j}$.
	Thus, there are exactly $A$ $\alpha$-type jobs on each machine from~$\mathcal{M}_j$.
	Observe that on a machine from~$\mathcal{M}_j$, we cannot use a job $\alpha_{j'}$, where $j' < j$, as this would exceed~$T + a_i$.
	Therefore, we have to execute $A$ jobs of type~$\alpha^{\times}_{k}$ on each machine from~$\mathcal{M}_{k}$.
	Thus, all jobs of type~$\alpha^{\times}_{k}$ have to be executed by machines in $\mathcal{M}_{k}$.
	Consequently, we have to execute $A$ jobs of type~$\alpha^{\times}_{k-1}$ on each machine of $\mathcal{M}_{k-1}$ since there are no more jobs of type~$\alpha^{\times}_{k}$ available.
	This argument inductively propagates for all $j=k, k-1, k-2, \dots, 1$.
	Hence, on each machine the remaining space is at most\footnote{This maximum can only be reached if there are $A$ jobs of type~$\alpha^0_{j}$ and no jobs of type~$\alpha^1_{j}$ on a machine.} $a_{\max} < A < p_t$ for any job type $t$, so no other job can be scheduled.
	Consider the sizes of the jobs that have to be executed on a machine.
	There can be at most $a_i$ jobs of type $\alpha^1_j$ on each machine~$M_i$.
	Hence we have, for each $j \in [k]$,
	\begin{equation}
	\label{eq:a-over-k-le-sum}
	A/k \le \sum_{M_i \in \mathcal{M}_j} a_i
	\end{equation}
	because all $A/k$ jobs of type~$\alpha^1_j$ are assigned to machines of~$\mathcal{M}_j$.
	Moreover, we have 
	\begin{equation*}
	\sum_{j \in [k]} \sum_{M_i \in \mathcal{M}_j} a_i = A \enspace .
	\end{equation*}
	So if there was a $j \in [k]$ with $A/k < \sum_{M_i \in \mathcal{M}_j} a_i$, then there would be a $j' \in [k]$  with $A/k > \sum_{M_i \in \mathcal{M}_{j'}} a_i$.
	Since this would contradict 
	Equation~(\ref{eq:a-over-k-le-sum}),
	we have
	\begin{equation*}
	\sum_{M_i \in \mathcal{M}_j} a_i = \frac{A}{k} = B
	\end{equation*}
	and $a_i$ jobs of type~$\alpha^1_j$ on each $M_i \in \mathcal{M}_j$ for each $j \in [k]$.
	Hence, $\sigma$ is perfect and the sets $\{a_i \mid M_i \in \mathcal{M}_j\}$ for each $j \in [k]$ are a solution for the corresponding instance of \balancedbinpacking.
\end{proof}

We can easily adjust our hardness instance~$I$ of $Q|HM|C_{\max}$ to an instance~$I_R$ of $R|HM|C_{\max}$.
Instead of machine speeds depending, for machine~$M_i$, on~$a_i$, we will use a larger makespan~$T_R$ to host a new ``blocker'' job type~$\gamma$, whose length is machine-dependent, and leaves space~$T + a_i$ on each machine---previously the capacity on a machine with speed~$s_i$.

\ifproc
\begin{lemma}[\appmark]
\else
\begin{lemma}
\fi
	\label{lem:reductionBBPtoRCmax}
	\balancedbinpacking with tight instances reduces to $R|HM|C_{\max}$ such that
	\begin{enumerate*}
		\item the number of machines equals the number of items,
		\item the number of job types equals $3k + 1$, where $k$ is the number of bins,
		\item the job sizes and job multiplicities are bounded by $\Oh(A^4)$, where $A$ is the sum of all items of the \balancedbinpacking instance,\label{lem:reductionBBPtoRCmax-item:numbersPolynomial}
		\item in any feasible schedule, all machines finish precisely by time~$T_R = 7 k A^3$, and
		\label{lem:reductionBBPtoRCmax-item:allMachinesReachMakespan}
		\item the job sizes matrix $\vep$ has rank $2$.
	\end{enumerate*}
\end{lemma}

\ifproc\else
\begin{proof}
	In the new hardness instance~$I_R$ for~$R|HM|C_{\max}$, we use the same job types with the same lengths and multiplicities as in~$I$, which is our hardness instance for~$Q|HM|C_{\max}$.
	We introduce a new job type~$\gamma$ with
	\begin{equation*}
	p^i_{\gamma} = 4 k A^3 - a_i, \qquad n_{\gamma} = m \enspace .
	\end{equation*}
	Observe that~$\gamma$ is the only job type that is machine-dependent.
	However its variation between machines is only~$- a_i$, which is relatively small compared to its total length.
	A perfect schedule for~$I_R$ is as a perfect schedule for~$I$, but with an additional job of type~$\gamma$ assigned once to each machine.
	Again, the parameters are clear from the definition of~$I_R$ and we prove the correctness next.
	
	On the one hand, if there is a solution~$\mathcal{S}$ of the corresponding instance of \balancedbinpacking,
	we construct a perfect scheduling for~$I$ as follows.
	If, in~$\mathcal{S}$, $a_i$ is assigned to the $j$-th bin,
	we assign to machine~$M_i$ $a_i$ jobs of type~$\alpha^1_j$, $A - a_i$ jobs of type~$\alpha^0_j$, one job of type~$\beta_j$, and one job of type~$\gamma$.
	This assignment has makespan~$T_R$ and, clearly, all jobs are assigned to some machine.
	
	On the one hand, assume that $I_R$ instance is feasible, meaning there is a schedule~$\sigma$ not exceeding the target makespan~$T_R$.
	Again, let us analyze the structure of such a solution.
	Per machine, there is exactly one job of type~$\gamma$ since we can execute at most one such job on each machine.
	The space remaining on machine~$M_i$ after executing a job of type~$\gamma$ is
	\begin{equation*}
	T_R - p^i_{\gamma} = 7 k A^3 - (4 k A^3 - a_i) = 3 k A^3 + a_i \enspace .
	\end{equation*}
	This is precisely the capacity of machine~$M_i$ in~$I$ as described in the proof of Lemma~\ref{lem:reductionBBPtoQCmax}.
	After scheduling all jobs of type~$\gamma$ there are also the same job types with the same lengths and multiplicities remaining.
	Thus, the rest of the analysis is the same.
	
	It remains to show that the rank of the job sizes matrix $\vep$ is 2.
	Define a matrix $C$ whose rows are indexed by the job types as follows.
	The row for job type $t \in \{ \alpha_j^0, \alpha_j^1, \beta_j \}$ (for every $j \in [k]$) is $(p_{t},0)$, and the row for $\gamma$ is $(4kA^3, -1)$.
	Next, define a matrix $D$ whose columns are indexed by the machines as follows: column $i \in [m]$ is $(1, a_i)$.
	It is easy to verify that $C \cdot D = \vep$.
\end{proof}
\fi

Applying the reductions of Lemmas~\ref{lem:reductionBBPtoQCmax} and~\ref{lem:reductionBBPtoRCmax} to \balancedbinpacking with $2$ bins,
we have that $Q|HM|C_{\max}$ and $R|HM|C_{\max}$ are \NPh with 6 and 7 job types, respectively.
$R|HM|C_{\max}$ can be reduced to 4 job types, and similar ideas can be used to improve the previously described reduction to only require $3k - 2$ job types.
\begin{reptheorem}{thm:QCmaxNPhard}
	$Q|HM|C_{\max}$ is \NPh already with 6 job types.
\end{reptheorem}

\ifproc
\begin{theorem}[\appmark]
\else
\begin{theorem}
\fi
	\label{thm:RCmaxNPhard}
	$R|HM|C_{\max}$ is \NPh already with 4 job types and with $\vep$ of rank $2$.
\end{theorem}
\ifproc\else
	\begin{proof}[Proof of Theorem~\ref{thm:RCmaxNPhard}]
		We will modify the reduction described in Lemma~\ref{lem:reductionBBPtoRCmax} to use only 4 types of jobs if the number of bins $k = 2$.
		First, we remove the job type~$\gamma$ to get to 6 different types of jobs.
		Recall that~$p^i_{\gamma} = 4 k A^3 - a_i$.
		For the $4 k A^3$, we will account for when adjusting the makespan and we add the $- a_i$ to the $\beta$-type jobs (now $p^i_{\beta_j} = 2 k A^3 - A^2 (k-j) - a_i$).
		Second, we blow up the makespan by a factor of~$A/(7k)$.
		So we have~$T = A^4$.
		
		To reduce to 5 different types of jobs, we remove all jobs of type~$\beta_1$.
		Still, we want $A$ times a job of type~$\alpha^{\times}_{1}$ on every machine of~$\mathcal{M}_1$.
		So its size will be around~$A^3$.
		To distinguish between~$\alpha^{1}_{1}$ and~$\alpha^{0}_{1}$ and get a dependency of machine~$M_i$ on item~$a_i$, we add~$A - a_i$ and subtract~$a_i$, respectively.
		So, for $i \in [m]$, we have
		\begin{align}
		&p^i_{\alpha^1_1} = A^3 + A - a_i &\textrm{ and }&
		&p^i_{\alpha^0_1} = A^3 - a_i \enspace .
		\end{align}
		As in the previous reduction, we can fit $a_i$ times $p^i_{\alpha^1_1}$ and $A - a_i$ times $p^i_{\alpha^0_1}$ to a machine, which then needs precisely the makespan~$T$.
		
		To reduce to 4 different types of jobs, we remove all jobs of type~$\alpha^0_2$ and we change the length of~$\alpha^1_2$ to
		\begin{align}
		p^i_{\alpha^1_2} = A^2
		\end{align}
		for all $i \in [m]$.
		We lengthen the job of type~$\beta_2$ to
		\begin{align}
		p^i_{\beta_2} = A^4 - a_i A^2 \, ,
		\end{align}
		which is the makespan~$T$ minus $a_i$ times $p^i_{\alpha^1_2}$.
		Note that the rank of $\vep$ is still just $2$: the rows of $C$ are $(A^3 + A,-1)$ for $\alpha_1^1$, $(A^3,-1)$ for $\alpha_1^0$, $(A^2,0)$ for $\alpha_2^1$, and $(A^4, -A^2)$ for $\beta_2$, and $D$ is defined as before.
		
		It remains to show the correctness of this reduction.
		Clearly, if there is a solution to the instance of \balancedbinpacking with 2 bins (i.e. a partition), we can assign the jobs to the machines as in the perfect schedule from Lemma~\ref{lem:reductionBBPtoQCmax} ignoring~$\beta_1$ and~$\alpha^0_2$.
		
		Assume there is a solution of the obtained instance of~$R|HM|C_{\max}$.
		On half of the machines, there is a job of type~$\beta_2$.
		On these machines, namely~$\mathcal{M}_2$, there is no space for a job of type~$\alpha^{\times}_{1}$.
		So, all $Am/2$ jobs of type~$\alpha^{\times}_{1}$ are scheduled to the~$m/2$ machines of~$\mathcal{M}_1$.
		As there cannot be more than~$A$ jobs of type~$\alpha^{\times}_{1}$ on a machine, there are precisely $A$ jobs of type~$\alpha^{\times}_{1}$ on each machine of~$\mathcal{M}_1$---at most~$a_i$ of which can be~$\alpha^{1}_{1}$.
		Thus, the free space on such a machine is at most~$a_{\max} < A$, so there is no job of type~$\alpha^1_2$ on these machines.
		To schedule all $A/2$ jobs of type~$\alpha^1_j$ for $j \in [2]$,
		we have to choose $\mathcal{M}_j$ such that the corresponding item sizes in the \balancedbinpacking instance sum up to at least~$A/2$.
		As the total sum of items is~$A$, both partitions correspond to items summing up to precisely~$A/2$.
		This yields a equal partition of the items.
	\end{proof}
\fi
The complexity of $Q|HM|C_{\max}$ ($R|HM|C_{\max}$) with less than $6$ ($4$) job types remains open.

From Lemmas~\ref{lem:reductionBBPtoQCmax} and~\ref{lem:reductionBBPtoRCmax} and the hardness of Corollary~\ref{cor:unarybalancedbinpacking-is-W1-hard}, we also get our main result:
\begin{reptheorem}{thm:xcmaxwh1}
	$X||C_{\max}$ is \Wh{1} parameterized by the number of job types with
	\begin{enumerate*}
		\item $X=Q$ and $\ven$, $\vep$, and $\ves$ given in unary.
		\item $X=R$ and $\ven$ and $\vep$ given in unary and $\rank(\vep) = 2$.
	\end{enumerate*}
\end{reptheorem}

\subsection{\NP-hardness of \textsc{Cutting Stock}}

\prob{\textsc{Cutting Stock}}
{$k$ item types of sizes $\vep = (p_1, \dots, p_k) \in \N^{k}$ and multiplicities $\ven = (n_1, \dots, n_k) \in \N^k$, $m$ bin types with sizes $\ves = (s_1, \dots, s_m) \in \N^m$ and costs $\vecc = (c_1, \dots, c_m) \in \N^m$.}
{A vector $\vex = (x_1, \dots, x_m) \in \N^m$ of how many bins to buy of each size, and a packing of items to those bins, such that the total cost $\vecc \vex$ is minimized.}

The difficulty in transferring hardness from $Q|HM|C_{\max}$ to \textsc{Cutting Stock} is in enforcing that each bin type is used exactly once.


\begin{lemma}
\label{lem:cuttingstock}
$Q|HM|C_{\max}$ with $k$ job types and $m$ machines reduces to \textsc{Cutting Stock} with $k+2$ item types and $m$ bin types.
\end{lemma}
	\begin{proof}
	We will set the sizes of bin types as $3$-dimensional vectors, whose interpretation as numbers is straightforward by choosing the base of each coordinate sufficiently large to prevent carry when summing.
	For machine~$M_i$ with capacity $T+a_i$, we add a bin type of size and cost $(1,2^{i - 1},T+a_i)$.
	For each original job type $t$ of size $p_t$, there is an item type of size $(0, 0, p_t)$ with the same multiplicity $n_t$.
	We will add two new item types: there are $m$ items of type $\eta$ which have size $(1, 0, 0)$, and $2^{m}-1$ items of type $\nu$ which have size $(0, 1, 0)$.
	The target cost is $C = (m, 2^m - 1, mT + A)$.
	
	Clearly, a feasible schedule translates easily to a packing: buying each bin type exactly once costs exactly $C$, the original item types are packed according to the feasible schedule, and we pack one $\eta$-type job and $2^{i - 1}$ $\nu$-type jobs on machine $M_i$.
	
	In the other direction, first notice that we have to use at least $m$ bins to pack the $\eta$-type jobs, and at most $m$ bins are affordable due to the budget $C$.
	We want to show that we have to use each bin type exactly once.
	Focus on the second coordinates of the $3$-dimensional vectors.
	Since the total size of items with respect to these coordinates is $2^m - 1$, which is precisely the affordable capacity, a solution to \textsc{Cutting Stock} must buy $m$ bins with capacity $2^m - 1$.
	This is equivalent to decomposing the number $2^{m}-1$ into a sum of some $m$ numbers which are powers of $2$, namely $2^0, 2^1, \dots, 2^{m-1}$.
	Clearly, the unique decomposition is $2^m-1 = 2^0 + 2^1 + \cdots + 2^{m-1}$.
	Hence, the unique way to obtain capacity $C$ by buying $m$ bins is to buy one bin of each type, concluding the proof.
	\end{proof}
Note that the \Whness{1} of $Q||C_{\max}$ does not immediately imply \Whness{1} of \textsc{Cutting Stock} when $\vep, \ven, \vecc$ are given in unary, because the construction of Lemma~\ref{lem:cuttingstock} blows up each of $\vep, \ven, \vecc$: it introduces large costs, items $\eta$ with large size, and items $\nu$ with large multiplicity.

Using our hardness of $Q|HM|C_{\max}$ with $6$ job types together with Lemma~\ref{lem:cuttingstock} yields:
\begin{reptheorem}{thm:cuttingstock}
\textsc{Cutting Stock} is \NPh already with $8$ item types.
\end{reptheorem}
%

\subsection{Hardness of $Q||\ell_2$ and $R||\ell_2$}

We will now transfer our hardness reduction to the $\ell_2$ norm.
Remember that the speed~$s_i$ of machine~$M_i$ depended linearly on~$T + a_i$ (normalized by~$1/T$ for all machines).
For the $\ell_2$ norm, we observe that
the machine speed affects the objective value by its square.
So for a machine where we double its speed, it contributes only a fourth to the objective value.
Then, one can construct an instance where it is more beneficial to schedule more than the loads of a perfect schedule to the faster machines leaving the slower machines rather empty.

To still apply our argument that the perfect schedules,
which precisely correspond to bin packings,
are the only ones admitting an optimal schedule,
we adjust the machine speeds.
It should be a value in the order of~$\sqrt{ T + a_i }$.
We use the ceiling function to have rational machine speeds.
However, for our reduction it is crucial that machines $M_i$ and $M_j$ have a different speed if $a_i \ne a_j$.
To make each $\left\lceil \sqrt{T + a_i} \, \right\rceil$ different from $\left\lceil \sqrt{T + a_i - 1} \, \right\rceil$, we scale up $\sqrt{T+a_i}$ by a sufficiently large factor.
We will see that we can set this factor to be $(T + a_{\max})$,
which results, for machine~$M_i$, in a new machine speed of
\begin{equation}
\label{eq:Q||ell2-machine-speed}
s_i = \left\lceil (T + a_{\max}) \sqrt{ T + a_i } \enspace \right\rceil \enspace .
\end{equation}

In the following we will use~$\ell_2^2$, which is the square of the $\ell_2$ norm, and is isotonic to it.
Recall that the unscaled load of $M_i$ is $\bar{L}_i = L_i \cdot s_i = \sum_{t=1}^{\tau} p_t^i x_t^i$, where $\vex^i = (x_1^i, \dots, x_\tau^i)$ is the vector of job multiplicities scheduled to machine~$M_i$, and $\tau$ is the number of job types.

\begin{lemma}
	\label{lem:reductionBBPtoQl2}
	The hardness instance $I$ with modified $s_i$ is also hard for $Q|HM|\ell^2_2$ with target value $\sum_{i = 1}^m \left( (T + a_i) / s_i \right)^2$.
\end{lemma}

\begin{proof}
	As before, if the instance of \balancedbinpacking has a solution
	where item~$a_i$ is assigned to the $j$-th bin,
	we construct a perfect schedule,
	where we assign $a_i$ jobs of type $\alpha^1_j$, $A - a_i$ jobs of type $\alpha^0_j$ and one job of type $\beta_j$ to machine~$M_i$ for each $i \in [m]$.
	As this gives us load~$(T + a_i) / s_i$ on machine~$M_i$, we reach precisely the target objective value~$\sum_{i = 1}^m \left( (T + a_i) / s_i \right)^2$ for the $\ell_2^2$ objective.
	
	For the other direction, assume there is a schedule~$\sigma$ of jobs to machines such that the objective value is at most $\sum_{i = 1}^m \left( (T + a_i) / s_i \right)^2$.
	We distinguish two cases.
	
	\noindent\textbf{Case 1:} \emph{The unscaled load of machine~$M_i$ is ~$T + a_i$, for each $i \in [m]$.}
	Observe that the objective value of~$\sigma$ equals the prescribed threshold objective value $\sum_{i = 1}^m \left( (T + a_i) / s_i \right)^2$.
	By Lemma~\ref{lem:reductionBBPtoQCmax}~\ref{lem:reductionBBPtoQCmax-item:allMachinesReachMakespan}, we know that such a schedule is perfect and exists if and only if there is a solution to the corresponding \balancedbinpacking instance.
	
	\noindent\textbf{Case 2:} \emph{There is an $i \in [m]$ such that $M_i$ has unscaled load different from $T+a_i$.}
	Consider the unscaled loads~$\LL = (\bar{L}_1, \dots, \bar{L}_m)$ scheduled to each of the machines in~$\sigma$.
	Since the total unscaled load is independent of the schedule, we can reach $\LL$ from the  ``perfect'' unscaled load distribution~$(T + a_1, \dots, T + a_m)$ of a perfect schedule (as it appears in Case~1) by iteratively moving a portion of the load from one machine to another.
	Note that we do not speak of moving jobs here.
	For this argument, we only consider the unscaled load of each machine as an integral number and ignore the jobs.
	In this process
	\begin{itemize}
		\item $m$ iterations of re-distribution are sufficient; in each step we take the machine with the smallest deviation (minimizing $\Delta_i = |\bar{L}_i - (T+a_i)|$) and move $\Delta_i$ integral units of load from it or to it (depending on the direction of the deviation).
		Note that there exists some other machine~$M_j$ to/from which to move because we chose $i$ to minimize $\Delta_i$.
		\item the load of each machine monotonously increases, decreases, or remains unchanged, i.e., we do not first add and then remove a portion of load or the other way around.
	\end{itemize}
	We show that in every step the objective value only increases, hence this case cannot occur as we already matched the threshold objective value in the ``perfect'' distribution of Case~1.
	
	Consider one such step.
	We move load $r \ge 1$ to machine~$M_i$ and take it from machine~$M_j$.
	Before, we have already moved in total~$z_i \ge 0$ to~$M_i$
	and we have already removed in total~$z_j \ge 0$ from~$M_j$.
	If~$M_i$ is slower than~$M_j$, then the objective value definitely increases.
	Hence, we assume $s_i \ge s_j$ (this implies $a_i \ge a_j$).
	So it remains to show
	\begin{align}	
	&& \left( \frac{T+a_i+z_i}{s_i} \right)^2 + \left( \frac{T+a_j-z_j}{s_j} \right)^2 &< \left( \frac{T+a_i+z_i+r}{s_i} \right)^2 + \left( \frac{T+a_j-z_j-r}{s_j} \right)^2 \nonumber \\ 
	&\Leftrightarrow& s_i^2 \left(2r (T + a_j - z_j) - r^2 \right) &< s_j^2 \left(2r (T + a_i + z_i) + r^2 \right) \nonumber \\
	&\Leftrightarrow& \frac{s_i^2}{s_j^2} &< \frac{2 (T + a_i + z_i) + r}{2 (T + a_j - z_j) - r} \enspace . \label{eq:Qell2-Case2}
	\end{align}
	
	Next, we analyze the machine speed $s_i$ as defined in equation~\eqref{eq:Q||ell2-machine-speed}.
	Recall that we scale up $\sqrt{T+a_i}$ by a sufficiently large factor~$b$
	to make each $\left\lceil \sqrt{T + a_i} \, \right\rceil$ different from $\left\lceil \sqrt{T + a_i - 1} \, \right\rceil$
	If the difference between $\sqrt{T+a_i}$ and $\sqrt{T+a_i-1}$ is at least $d$, then it must hold that
	\begin{equation*}
	b > \frac{1}{d} \ge \frac{1}{\sqrt{T + a_{\max}} - \sqrt{T + a_{\max} - 1}} \enspace .
	\end{equation*}
	We have chosen $b = T + a_{\max}$, since $x > 1 / (\sqrt{x} - \sqrt{x - 1})$ for $x \ge 4$.
	Hence, we conclude
	\begin{equation}
	\label{eq:Q||ell2-Case2-speedUpperBound}
	\left\lceil (T + a_{\max}) \sqrt{ T + a_i } \, \right\rceil < (T + a_{\max}) \sqrt{ T + a_i + 1} \enspace .
	\end{equation}
	With this inequality in hand, we finally show the correctness of inequality~\eqref{eq:Qell2-Case2}:
	\begin{align*}
	\frac{s_i^2}{s_j^2}
	&= \frac{\left\lceil (T + a_{\max}) \sqrt{ T + a_i } \, \right\rceil^2}{\left\lceil (T + a_{\max}) \sqrt{ T + a_j } \, \right\rceil^2}
	\enspace \overset{\eqref{eq:Q||ell2-Case2-speedUpperBound}}{<} \enspace \frac{(T + a_{\max})^2 (T + a_i + 1)}{(T + a_{\max})^2 (T + a_j)}
	= \frac{2 (T + a_i) + 2}{2 (T + a_j)} \\
	& \overset{(r \ge 1)}{\le} \enspace \frac{2 (T + a_i) + 2 r}{2 (T + a_j)}
	\enspace \overset{(a_i \ge a_j)}{<} \enspace \frac{2 (T + a_i) + r}{2 (T + a_j) - r}
	\le \frac{2 (T + a_i + z_i) + r}{2 (T + a_j - z_j) - r} \qedhere
	\end{align*}
\end{proof}

Similarly, we can transfer our hardness instance to $R|HM|\ell^2_2$.

\begin{lemma}
	\label{lem:reductionBBPtoRl2}
	The hardness instance $I_R$ is hard for $R|HM|\ell^2_2$ with target value $m\cdot T_R^2$.
\end{lemma}
\begin{proof}
	Again, if the instance of \balancedbinpacking has a solution
	where item~$a_i$ is assigned to the $j$-th bin,
	we construct a perfect schedule,
	where we schedule $a_i$ jobs of type~$\alpha^1_j$, $A - a_i$ jobs of type~$\alpha^0_j$, one job of type~$\beta_j$, and one job of type~$\gamma$ to machine~$M_i$ for each $i \in [m]$.
	As this gives us processing time~$T_R$ per machine, we precisely reach the target objective value of~$m T_R^2$ for the $\ell_2^2$ objective.
	
	For the other direction, assume there is a schedule of jobs to machines such that the objective value is at most $mT_R^2 = 49 m k^2 A^6$.
	We distinguish three cases.
	
	\noindent\textbf{Case 1:} \emph{The load of each machine is at most~$T_R = 7 k A^3$.}
	Such a schedule would thus have makespan $T_R$ and is feasible for $R|HM|C_{\max}$ with target makespan $T_R$.
	By Lemma~\ref{lem:reductionBBPtoRCmax}, we know that such a schedule exists if and only if there is a solution to the corresponding \balancedbinpacking instance.
	By property~\ref{lem:reductionBBPtoRCmax-item:allMachinesReachMakespan} of Lemma~\ref{lem:reductionBBPtoRCmax}, it admits an objective value of precisely~$m T_R^2$ for the $\ell_2^2$ objective.
	
	\noindent\textbf{Case 2a:} \emph{There is a machine with load~$T_R' > T_R = 7 k A^3$, and on each machine there is precisely one job of type~$\gamma$.}
	Since the processing time for all $\alpha$- and $\beta$-type jobs is the same on all machines and we have exactly one job of type~$\gamma$ per machine,
	the total load is independent of the schedule and is $m \cdot T_R$.
	Fixing the total load, the $\ell_2^2$ objective reaches its minimum uniquely by distributing the load evenly; see e.g.~\cite[Proof of Theorem 3]{KnopK:2017}.
	Thus, the objective $m T_R^2$ can only be reached if the load of every machine is $T_R$, so this case cannot occur.
	
	\noindent\textbf{Case 2b:} \emph{There is a machine which schedules at least two jobs of type~$\gamma$.}
	In this case, we exploit Claim~\ref{claim:gapCase2b}, which we prove \ifproc in the full version. \else next. \fi
	Again, it contradicts our assumption of $\sigma$ having objective value at most $mT_R^2$. So this case can also not occur.
	\ifproc
	\begin{claim}[\appmark]
	\else
	\begin{claim}
	\fi
		\label{claim:gapCase2b}
		Any schedule in Case~2b has objective value strictly greater than $r \cdot mT_R^2$ with $r = (m - 0.98) / (m - 1)$.
		Hence, the objective value of such a schedule exceeds~$mT_R^2$ by at least
		\begin{align*}
		(r-1) m T_R^2 = \frac{0.02}{m - 1} \cdot 49 m k^2 A^6
		> 0.98 k^2 A^6 \enspace .
		\end{align*}
	\end{claim}

\ifproc
\end{proof}
\else
	\medbreak
	
	\noindent \textit{Proof:}

	The dependence of~$p^i_{\gamma}$ on the choice of a machine~$M_i$ is only subtracting~$a_i$.
	So we get a lower bound on the total sum of job sizes of all jobs in any schedule if we subtract $m$ times $a_{\max}$ (as we have $m$ jobs of type $\gamma$).
	This yields a total sum
	\begin{align*}
	\mathcal{T} =& \sum_{j \in [k]} \left( n_{\alpha^1_j} p_{\alpha^1_j} + n_{\alpha^0_j} p_{\alpha^0_j} + n_{\beta_j} p_{\beta_j} \right) + m (4 k A^3 - a_{\max}) \nonumber \\
	=& 7 m k A^2 + A - m a_{\max} \enspace .
	\end{align*}
	The machine where we have scheduled two jobs of type $\gamma$ has load at least
	\begin{equation*}
	T_R' \ge 2 \cdot (4 k A^3 - a_{\max}) > 8kA^3 - 2A \enspace .
	\end{equation*}
	This is already greater than~$T_R$, which is in turn at least $\mathcal{T} / m$.
	Hence, we assume for the rest of the proof that $T_R'$ is exactly $8kA^3 - 2A$ and the remaining processing time $\mathcal{T} - T_R'$ is distributed equally across the other $m - 1$ machines as otherwise the resulting objective value would only increase.
	The average load $L_{\text{avg}}$ of the remaining machines is 
	\begin{align*}
	L_{\text{avg}} = \frac{\mathcal{T} - T_R'}{m - 1} = \frac{(7 m k A^3 + A - m a_{\max}) - (8kA^3 - 2A)}{m - 1}
	> \frac{7mkA^3 - 8kA^3 - mA}{m-1} \enspace .
	\end{align*}
	Hence, the objective value of such a schedule is at least
	\begin{gather*}
	\left(8 k A^3 - 2A\right)^2 + (m-1) \left(\frac{7mkA^3 - 8kA^3 - mA}{m-1}\right)^2 \\
	> \frac{64 mk^2A^6 - 64 k^2A^6 - 32 mkA^4 + 49 m^2k^2A^6 - 112 mk^2A^6 - 14 m^2kA^4}{m-1} \\
	> 49 mk^2A^6 \frac{m - \frac{48}{49} - \frac{64}{m} - \frac{46m}{A}}{m-1} 
	\ge r mT_R^2 \enspace , \\
	\textrm{\hspace{-16pt} where \hspace{10pt}}
	\frac{m - \frac{48}{49} - \frac{64}{m} - \frac{46m}{A}}{m-1} \ge r = \frac{m - \frac{49}{50}}{m - 1}
	\end{gather*}
	because, without loss of generality, we can assume that $m \ge 64 \cdot 4900$ as otherwise we could add more dummy items to our \balancedbinpacking as described in the proof of Lemma~\ref{lem:reductionBPtoBBP},
	and we can assume that $A \ge 46 \cdot 4900 m$ as otherwise we could scale up the items of the \balancedbinpacking instance by a factor of~$46 \cdot 4900$.
\end{proof}
\fi

The following corollaries follow immediately from Lemmas~\ref{lem:reductionBBPtoQl2} and~\ref{lem:reductionBBPtoRl2};
as before, it is likely that one might improve this to 4 job types.

\begin{corollary}
	\label{cor:Xl2NPhard}
	$X|HM|\ell_2$ is \NPh already for $t$ job types with
	\begin{enumerate*}
		\item $X=Q$, $t = 6$.
		\item $X=R$, $t = 7$, and $\rank(\vep) = 2$.
	\end{enumerate*}
\end{corollary}
\begin{corollary}
	\label{cor:Xl2W1hard}
	$X||\ell_2$ is \Wh{1} parameterized by the number of job types with
	\begin{enumerate*}
		\item $X=Q$ and $\ven$, $\vep$, and $\ves$ given in unary.
		\item $X=R$ and $\ven$ and $\vep$ given in unary and $\rank(\vep) = 2$.
	\end{enumerate*}
\end{corollary}

\subsection{Hardness of $R||\sum w_j C_j$}
We will define weights in the hardness instance~$I_R$ from Lemma~\ref{lem:reductionBBPtoRCmax}.
Denote $\rho^i_j = w_j / p^i_j$ the \emph{Smith ratio} of a job $j$ on machine~$M_i$, where $w_j$ is its weight.
It is known that given an assignment of jobs to machines, an optimal schedule is obtained by executing jobs ordered by their Smith ratios (on each machine) non-increasingly~\cite{Smith1956}.
It suffices to restrict ourselves to such schedules, and an assignment of jobs to machines describes such a schedule.

We would like to use the same approach as for $\ell_2$ (Lemma~\ref{lem:reductionBBPtoRl2}) because it is known that $\sum w_j C_j$ and $\ell_2$ are often (not always) closely related.
However, because the size of a job of type $\gamma$ depends both on~$j$ and the machine~$M_i$,
yet its weight only depends on~$j$, it is impossible to express an exact objective value of the perfect schedule from the previous sections.
This would make the argument of an analogue of Case~2a of Lemma~\ref{lem:reductionBBPtoRl2} invalid and a no-instance of \balancedbinpacking might reduce to a yes-instance of $R||\sum w_j C_j$.
The contribution of all $\alpha$- and $\beta$-type jobs to the sum of weighted completion times is always the same as they and their weights are machine-independent.
However, the contribution of jobs of type~$\gamma$ depends on the machine, while its weight is machine-independent.
If we schedule to each machine exactly one job of type~$\gamma$,
then we will have each machine-dependent processing time once and across all machines their contribution is independent of the schedule and we can specify an exact target objective value.
Consequently, we can apply the same argumentation for Case~1 and Case~2a as in Lemma~\ref{lem:reductionBBPtoRl2}.
For Case~2b, we will exploit the claim in the proof of Lemma~\ref{lem:reductionBBPtoRl2} once again and combine it with a gap argument (Lemma~\ref{lem:gapsumwjCj}).

To obtain the \emph{weighted hardness instance} $I^w_R$, we define the following weights for our hardness instance~$I_R$ from Section~\ref{sec:hardnessX||Cmax}.
For the $\alpha$- and $\beta$-type jobs the weight equals its processing time and for the job type~$\gamma$ it is slightly greater:
\begin{equation*}
w_{\alpha^{\times}_j} = p_{\alpha^{\times}_j} \qquad w_{\beta_j} = p_{\beta_j} \qquad w_{\gamma} = 4 k A^3 \, (= p^i_{\gamma} + a_i \textrm{ for each } i \in [m])
\end{equation*}

\ifproc
\begin{lemma}[\appmark]
\else
\begin{lemma}
\fi  
\label{lem:gapsumwjCj}
Let $\sigma$ be any schedule of the weighted hardness instance,
let $(L_1, L_2, \dots, L_m)$ be its load vector, and $\LL := \frac{1}{2}\left(\sum_{i=1}^m {L_i}^2\right)$.
Let $\Gamma = \frac{1}{2k} \sum_{j=1}^k \left(A w_{\alpha^0_j}^2 + (m-1)A w_{\alpha^1_j}^2 + m w_{\beta_j}^2 \right)$, \\
$\Delta^{1:1}_{\operatorname{linear}} = \frac{1}{2} \sum_{i=1}^m p^i_{\gamma} w_{\gamma}$,
$\Delta^{1:1}_{\operatorname{quadr}} = \frac{1}{2} \sum_{i=1}^m p^i_{\gamma} \cdot a_i$,
$\Delta^{1:1} = \Delta^{1:1}_{\operatorname{linear}} + \Delta^{1:1}_{\operatorname{quadr}}$, \\
$\Delta^{\min}_{\operatorname{linear}} = m (w_{\gamma} - a_{\max}) w_{\gamma}$, and
$\Delta^{\min}_{\operatorname{quadr}} = m (w_{\gamma} - a_{\max}) a_{\max}$ \, . \\
\begin{enumerate*}
	\item The value of $\sigma$ under $\sum w_j C_j$ is at least $\LL + \Gamma + \Delta^{\min}_{\operatorname{linear}} + \Delta^{\min}_{\operatorname{quadr}}$. \label{lem:gapsumwjCj-minimumValue}\hspace{80pt}
	\item If $\sigma$ schedules one $\gamma$ job per machine,
	then the value $\sigma$ under $\sum w_j C_j$ is $\LL + \Gamma + \Delta^{1:1}$.
\end{enumerate*}
\end{lemma}
\ifproc\else
\begin{proof}[Proof of Lemma~\ref{lem:gapsumwjCj}]
	First notice that the Smith ratio of all $\alpha$- and $\beta$-type jobs is $1$, and the Smith ratio of the jobs of type~$\gamma$ is strictly greater than $1$, so the jobs of type~$\gamma$ will always be executed first.
	We use the following description of the objective function due to Knop and Koutecký~\cite{KnopK:2017}.
	Assume that $\tau$ job types are ordered according to their Smith ratios with respect to some machine~$M_i$ (with $i \in [m]$) as $t=1, \dots, \tau$, $x_t^i$ is the number of jobs of type $t$ scheduled on machine~$M_i$, and $z_t^i = \sum_{\ell=1}^t p_\ell^i x_\ell^i$ is the time spent processing the first $t$ job types.
	Define $\rho_{\tau+1}^i = 0$.
	Then the contribution of machine $M_i$ to the total $\sum w_j C_j$ objective is
	\begin{equation*}
	\frac{1}{2}\sum_{t=1}^\tau \left[ \left(z_t^i)^2(\rho_t^i - \rho_{t+1}^i\right) + p_t^i w_t x_t^i\right] \enspace .
	\end{equation*}
	In our case, the coefficients of $(z_t^i)^2$ for any $\alpha$- and $\beta$-type except the last one will be $0$ because their slopes are identical, hence $\rho_t^i - \rho_{t+1}^i = 0$.
	The term of the last $\alpha$- or $\beta$-type will have $z_t^i = L_i$ be the load of machine $M_i$ and its coefficient is $\rho^i_\tau - \rho^i_{\tau+1} = 1 - 0$, so this term is $\frac{1}{2} L_i^2$.
	Hence, subtracting those terms over all machines gives $\LL$, and we are left to account for%
	\begin{enumerate*} \item the quadratic terms corresponding to the jobs of type~$\gamma$, and \item the linear terms $p_t^i w_t x_t^i$.\end{enumerate*}
	
	First, we consider the linear terms for the $\alpha$- and $\beta$-type jobs.
	Since the sizes of these jobs are independent of the machines, we just sum them up without knowing to which machine they are scheduled.
	For each $j \in [k]$, we have $A/k$ jobs of type~$\alpha^1_j$, $(m-1)A/k$ jobs of type~$\alpha^0_j$ jobs and $m/k$ job of type~$\beta$.
	Hence, across all $j \in [k]$ this is
	\begin{align*}
	& \frac{1}{2} \sum_{j=1}^k \left( \frac{A}{k} \cdot p_{\alpha^0_j} w_{\alpha^0_j} + \frac{(m-1) A}{k} \cdot p_{\alpha^1_j} w_{\alpha^1_j} + \frac{m}{k}\cdot p_{\beta_j} w_{\beta_j} \right) \\
	=& \frac{1}{2k} \sum_{j=1}^k \left(A w_{\alpha^0_j}^2 + (m-1)A w_{\alpha^1_j}^2 + m w_{\beta_j}^2 \right) \\
	=& \Gamma \, .
	\end{align*}
	
	Now, we consider the jobs of type~$\gamma$.
	Let us first assume that we have scheduled exactly one job of type~$\gamma$ per machine.
	This means that across all machines, every possible quadratic and linear term appears precisely once.
	So for the linear terms, we get
	\begin{equation*}
	\frac{1}{2} \sum_{i=1}^m p^i_{\gamma} w_{\gamma} = \Delta^{1:1}_{\operatorname{linear}} \enspace .
	\end{equation*}
	For the quadratic terms, we get
	\begin{equation*}
	\frac{1}{2} \sum_{i=1}^m (p^i_{\gamma})^2 \left(\frac{w_{\gamma}}{p^i_{\gamma}} - 1 \right)
	= \frac{1}{2} \sum_{i=1}^m p^i_{\gamma} \left(w_{\gamma} - p^i_{\gamma}\right)
	= \frac{1}{2} \sum_{i=1}^m p^i_{\gamma} \cdot a_i
	= \Delta^{1:1}_{\operatorname{quadr}} \enspace .
	\end{equation*}
	
	Let us now drop the assumption that we have scheduled exactly one job of type~$\gamma$ per machine and determine a lower bound for the objective value of an arbitrary schedule.
	Still, $\LL$ and $\Gamma$ have the structure described above. 
	Thus, we specify a lower bound by minimizing the linear and the quadratic terms for the jobs of type~$\gamma$.
	Clearly, they are minimum if we schedule each of the $m$ jobs to the machine where it has the smallest size---this is machine~$M_{i}$ corresponding to item~$a_{\max}$.
	Consequently, we have (since $p^i_{\gamma} = w_{\gamma} - a_{\max}$)
	\begin{equation*}
	\frac{1}{2} m (w_{\gamma} - a_{\max}) w_{\gamma}
	= \Delta^{\min}_{\operatorname{linear}}
	\end{equation*}
	and
	\begin{equation*}
	\frac{1}{2} m (w_{\gamma} - a_{\max}) a_{\max}
	= \Delta^{\min}_{\operatorname{quadr}} \enspace .
	\end{equation*}
\end{proof}
\fi

With this lemma at hand, it is not difficult to show that the weighted hardness instance indeed reduces \balancedbinpacking to $R|HM|\sum w_j C_j$ as before:
\ifproc
\begin{lemma}[\appmark]
\else
\begin{lemma}
\fi
	\label{lem:reductionBBPtoRsumwjCj}
	The weighted hardness instance~$I^w_R$ is hard for $R|HM|\sum w_j C_j$.
\end{lemma}
\ifproc\else
\begin{proof}[Proof of Lemma~\ref{lem:reductionBBPtoRsumwjCj}]
Set the target objective value to be $\frac{1}{2}kmT_R^2 + \Gamma + \Delta^{1:1}$.
Note that a perfect schedule satisfies the condition that every machine executes exactly one job of type $\gamma$ and the load of every machine is at most $T_R$, hence by Lemma~\ref{lem:gapsumwjCj}, the value of a perfect schedule is precisely the target value.
Assume a schedule $\sigma$ is given whose $\sum w_j C_j$ objective is at most the target objective.
We again distinguish three cases:

\noindent\textbf{Case 1:} \emph{The load of each machine is at most $T_R$.}
This is again a schedule of makespan at most $T_R$ and the analysis of Lemma~\ref{lem:reductionBBPtoRCmax} applies.
Hence, $\sigma$ is a perfect schedule.

\noindent\textbf{Case 2a:} \emph{Each machine contains exactly one $\gamma$-type job and there is a machine with load more than $T_R$.}
By Lemma~\ref{lem:gapsumwjCj}, we know that such a schedule has an objective value of~$\LL + \Gamma + \Delta^{1:1}$.
In this sum, $\Gamma$ and $\Delta^{1:1}$ are constant and independent of the loads of the machines.
We use the same argument as in Lemma~\ref{lem:reductionBBPtoRl2}.
As the objective value of~$\frac{1}{2}kmT_R^2 + \Gamma + \Delta^{1:1}$ is matched precisely if the total load is distributed evenly (i.e. Case~1),
re-distributing the same total load unevenly increases the quadratic term $\LL$.
Hence, this case cannot occur.

\noindent\textbf{Case 2b:} \emph{There is a machine which schedules at least $2$ jobs of type $\gamma$.}
By Lemma~\ref{lem:gapsumwjCj}\ref{lem:gapsumwjCj-minimumValue}, the objective value of~$\sigma$ is at least $\LL + \Gamma + \Delta^{\min}_{\operatorname{linear}} + \Delta^{\min}_{\operatorname{quadr}}$.
Let's compare this to the target value $\frac{1}{2}kmT_R^2 + \Gamma + \Delta^{1:1}_{\operatorname{linear}} + \Delta^{1:1}_{\operatorname{quadr}}$ summand by summand.
As shown in the proof of Lemma~\ref{lem:reductionBBPtoRl2}, in Case~2b we have $(\sum_{i=1}^{m} L_i^2) - mT_R^2 \geq (r-1) mT_R^2$.
However, we now have $\LL = \frac{1}{2} \sum_{i=1}^m {L_i}^2$.
\begin{equation*}
\textrm{\hspace{-16pt} Plugging in, we get} \enspace
\LL - \frac{1}{2} mT_R^2 \geq \frac{1}{2} (r-1) mT_R^2
> 0.49 k^2 A^6 \, .
\end{equation*}
Of course, $\Gamma$ is the same in both sums.
Consider each of the $km$ summands of $\Delta^{1:1}_{\operatorname{linear}}$.
Compared to its counterpart in $\Delta^{\min}_{\operatorname{linear}}$, it is greater by at most $w_{\gamma} a_{\max} < 4 k A^4$.
Similarly, consider each of the $km$ summands of $\Delta^{1:1}_{\operatorname{quadr}}$.
Compared to its counterpart in $\Delta^{\min}_{\operatorname{quadr}}$, it is greater by at most $a_{\max}^2 < A^2$.
Combining all $m$ summands of both of these sums, we have at most $5 m k A^4$.
This in turn is at most $0.25 k A^5$ because without loss of generality, we can assume that $A > 20 m$ as otherwise we could scale up the items of the \balancedbinpacking instance by a factor of~$20$.

So in total, the value of $\sigma$ is greater by at least $0.49 k^2 A^6$ minus at most $0.25 k A^5$, and thus cannot attain the target objective value, so this case also does not occur.
\end{proof}
\fi

\begin{corollary}
	\label{cor:RweightedSumNPhard}
	$R|HM|\sum w_j C_j$ is \NPh already with 7 job types and $\rank(\vep)=2$.
\end{corollary}
\begin{corollary}
	\label{cor:RweightedSumW1hard}
	$R||\sum w_j C_j$ is \Wh{1} parameterized by the number of job types, even if $\ven$ and $\vep$ are given in unary and $\rank(\vep)=2$.
\end{corollary}


\section{Open Problems}
We conclude with a few interesting questions raised by our results:
\begin{itemize}
\item We have shown that $Q|HM|C_{\max}$ and $R|HM|C_{\max}$ are \NPh with 6 and 4 job types, respectively. What is the complexity for smaller numbers of job types?
We are not aware of any positive result about either problem, including \textsc{Cutting Stock}, even for $2$ job/item types. 
\item Recall the question whether $P|HM|C_{\max}$ parameterized by the number of job types $k$ is in \FPT or not.
Our results provide some guidance for how one could use the interplay of high multiplicity of jobs and large job sizes to show hardness.
\item Is \textsc{Cutting Stock} \Wh{1} when the input data is given in unary?
\item We haven't yet investigated jobs with release times and due dates and minimization of makespan, weighted flow time, or weighted tardiness, already on one machine.
The work of Knop et al.~\cite{KnopKLMO:2019} shows that for example $1|r_j,d_j|\{C_{\max},\sum w_j F_j, \sum w_j T_j\}$ parameterized by the number of job types $k$ is in \XP when $p_{\max}$ is polynomially bounded.
Is it \FPT or \Wh{1}?
\end{itemize}

\bibliography{scheduling}

\end{document}